\documentclass[11pt]{article}
\usepackage{graphics,graphicx}
\usepackage{amssymb,amsmath,amsthm}
\usepackage{fullpage}
\usepackage{enumitem}
\usepackage{xcolor}
\usepackage{float}
\usepackage{hyperref}
\usepackage[capitalize]{cleveref}
\usepackage[margin=0.8in]{geometry}
\newtheorem{theorem}{Theorem}
\newtheorem{lemma}{Lemma}
\newtheorem{quest}{Question}
\newtheorem{prob}{Problem}

\newtheorem{coro}{Corollary}
\newtheorem{obs}{Observation}

\newtheorem{conj}{Conjecture}

\newtheorem{claim}{Claim}[lemma]

\usepackage[normalem]{ulem}
\usepackage{enumitem}
\usepackage{etoolbox}

\numberwithin{equation}{section}

\hypersetup{
	colorlinks=true,
	linkcolor=blue,
	filecolor=magenta,
	urlcolor=blue,
	pdftitle={Overleaf Example},
	pdfpagemode=FullScreen,
}

\newcommand{\Mod}[1]{\ (\mathrm{mod}\ #1)}

\usepackage{xpatch}
\makeatletter
\AtBeginDocument{\xpatchcmd{\cref@thmoptarg}{\thm@headpunct{.}}{\thm@headpunct{}}{}{}}
\AtBeginDocument{\xpatchcmd{\cref@thmnoarg}{\thm@headpunct{.}}{\thm@headpunct{}}{}{}}
\makeatother

\begin{document}
	
	\title{$4K_1$-free graph with the cop number $3$}
	
	\author{$^1$Arnab Char, $^{2}$Paras Vinubhai Maniya, and $^2$Dinabandhu Pradhan\thanks{Corresponding author.} \thanks{Supported in part by Core Research Grant, ANRF, India (CRG/2023/003749)}   \\ \\
		$^{1}$National Institute of Informatics \\
		Tokyo, Japan\\
		\small \tt Email: arnabchar@gmail.com
		\\ \\
		$^{2}$Department of Mathematics \& Computing\\
		Indian Institute of Technology (ISM) \\
		Dhanbad, India \\
		\small \tt Email: maniyaparas9999@gmail.com \\
		\small \tt Email: dina@iitism.ac.in}

	\date{}
	\maketitle

%
%

	\begin{abstract}
		The game of cops and robber is a two-player turn-based game played on a graph where the cops try to capture the robber. The cop number of a graph $G$, denoted by $c(G)$ is the minimum number of cops required to capture the robber. For a given class of graphs ${\cal F}$, let $c({\cal F}):=\sup\{c(F)|F\in {\cal F}\}$, and let Forb$({\cal F})$ denote the class of ${\cal F}$-free graphs.  We show that the complement of the Shrikhande graph is $(4K_1,C_{\ell}$)-free for any $\ell \geq 6$ and has the cop number~$3$. This provides a counterexample for the conjecture proposed by Sivaraman~\cite{Sivaraman19_1} which states that if $G$ is $C_{\ell}$-free for all $\ell\ge 6$, then $c(G)\le 2$. This also gives a negative answer to the question posed by Turcotte~\cite{Turcotte22} to check whether  $c($Forb$(pK_1))=p-2$ for $p\geq 4$.  Turcotte also posed the question to check whether $c($Forb$(pK_1+K_2))\leq p+1$ for $p\geq 3$. We prove that this result indeed holds. We also generalize this result for Forb$(pK_1+qK_2)$. Motivated by the results of Baird et al.~\cite{Baird14} and Turcotte and Yvon~\cite{Turcotte21}, we define the upper threshold degree and lower threshold degree for a particular class of graphs and show some computational advantage to find the cop number using these.
	\end{abstract}
	

	\section{Introduction}
	We consider finite, simple, undirected, and connected graphs in this paper. For a given graph $G$, let $V(G)$ and $E(G)$ denote the vertex set and the edge set of $G$, respectively. The complement of a graph $G$, denoted as $\overline{G}$, is the graph with the vertex set $V (\overline{G})=V(G)$ and the edge set $E(\overline{G})=\{uv\mid uv\notin E(G)\}$. Let $X\subseteq V(G)$. The subgraph of $G$ induced by the vertex set $X$ is denoted by $G[X]$, and the subgraph of $G$ induced by the vertex set $V(G)\setminus X$ is denoted by $G\setminus X$. The set $X$ is called a clique (independent set) if every pair of vertices of $X$  are adjacent (nonadjacent) in $G$.	For any two nonempty sets, say $X, Y\subseteq V(G)$, $X$ is complete (anticomplete) to $Y$ if every vertex in $X$ is adjacent (nonadjacent) to every vertex in $Y$. The open neighborhood of a vertex $x$ in $G$, denoted by $N_G(x)$, is the set of all adjacent vertices to $x$. The closed neighborhood of $x$ in $G$, denoted by $N_G[x]$, is the set $N_G(x)\cup \{x\}$. The \emph{degree} of a vertex $x$ in $G$ is defined as $|N_G(x)|$ and is denoted by $d_G(x)$. We denote $\delta(G)$ as the \emph{minimum degree} of $G$, where $\delta(G)=\min\{d_G(x) \mid x\in V(G)\}$.  We denote $\Delta(G)$ as the \emph{maximum degree} of $G$, where $\Delta(G)=\max\{d_G(x) \mid x\in V(G)\}$. We drop the subscript $G$ when there is no ambiguity.
	
	Let $C_\ell$, $K_\ell$, and $P_\ell$ denote the induced cycle, the induced complete graph, and the induced path on $\ell$ vertices, where $\ell \geq 1$. To avoid ambiguity, we  will write $K_i$ instead of $P_i$ for $i\in \{1,2\}$. If $G_1$ and $G_2$ are two vertex disjoint graphs, then the disjoint union of $G_1$ and $G_2$, denoted by $G_1+G_2$, is the graph with vertex set $V(G_1)\cup V(G_2)$ and edge set $E(G_1)\cup E(G_2)$. For a given graph $G$ and $\ell\in  \mathbb{N}$, we denote $\ell G$ as the disjoint union of $\ell$ copies of $G$. So $\ell K_1$ is the graph obtained from the disjoint union of $\ell$  copies of $K_1$. For a given graph $G$, we say $G$ is ${\cal F}$-free if for any $F\in {\cal F}$, $F$ is not an induced subgraph of $G$. For a given class of graphs ${\cal G}$, we say ${\cal G}$ is ${\cal F}$-free if for all $G\in {\cal G}$, $G$ is ${\cal F}$-free. If ${\cal F}=\{F\}$, we simply say $G$ is $F$-free, and ${\cal G}$ is $F$-free. We denote the class of ${\cal F}$-free graphs as Forb$({\cal F})$. 
	
	The game of cops and robber is a well studied problem in graph theory as can be seen in literature (see~\cite{Aigner84, Gupta23, Petr23, Turcotte21}). Quilliot~\cite{Quilliot78} introduced the game of cops and robber in 1978 and later, Nowakowski and Winkler~\cite{Nowakowski83} studied the above game independently. The game is a turn-based game and is played on a connected graph by two players. For a given graph $G$, player one controls the movement of a set of cops, say $D=\{D_1,D_2,\ldots,D_p\}$, and  player two controls the movement of a robber, say $R$, by putting them on vertices of $G$. A round of the game consists of two consecutive turns started with the turn of the player one and ended with the turn of the player two. The game is started by the player one by putting the cops in the vertices of $G$ and in the next turn, the player two puts the robber in a vertex of $G$. Clearly, this is round $1$ of the game. After that, the game is played on by alternative turns of the player one and the player two.  For $r\in \mathbb{N}$ and $i\in \{1,2,\ldots,p\}$, let $p_r^G(D_i)$  be the vertex, where the cop $D_i$ is placed at the round $r$ of the game in $G$. For $r\in \mathbb{N}$, let $p_r^G(R)$ be the vertex, where the robber $R$ is placed at the round $r$ in $G$. We drop the superscript $G$ when there is no ambiguity. For each $r\geq 2$, a move played by the player one in the round $r$ is legal if the player one puts $L$ in $p_{r-1}(L)$ or in some adjacent vertex of $p_{r-1}(L)$, where $L\in D$. For each $r\geq 2$, a move played by the player two in the round $r$ is legal if the player two puts $R$ in $p_{r-1}(R)$ or in some adjacent vertex of $p_{r-1}(R)$.  A cop \emph{captures} the robber if they are at the same vertex of the graph. In other words, the robber $R$ is captured by the set of cops $D$ in the round $r$ of the game if there exists a cop $L \in D$ such that $p_{r+1}(L)=p_{r}(R)$. This occurs only if, during the robber’s turn in the round $r$, there exists a cop $L \in D$ such that $p_{r}(R)\in N[p_{r}(L)]$, that is, the robber $R$ moved to a vertex in the closed neighborhood of a cop $L$.  In this case, we say $L$ captures $R$ in the round $r$. The player one wins the game if the robber is captured by the set of cops in a finite round of the game; otherwise the player two wins the game. Our goal is to find a winning strategy for the player one.
	
	\subsection{Existing results and problems}
	
	For a given graph $G$, the \emph{cop number} of $G$, denoted by $c(G)$~\cite{Aigner84}, is the minimum cardinality of the set of cops such that the robber is captured by that set of cops.  It is easy to verify that for a complete graph $K$, $c(K)=1$ and for a cycle $C_\ell$ with $\ell\geq 4$, $c(C_\ell)=2$. Researchers are always intrigued about the existence of a strategy to capture the robber in an arbitrary graph $G$. For a given class of graphs, say ${\cal G}$, we define  $c({\cal G}):= \sup\{c(G)| G\in {\cal G}\}$; see \cite{Das22}. Joret et al.~\cite{Joret10} showed that for a given graph $F$, $c($Forb$(F))$ is finite if and only if  every component of $F$ is a path. This clearly implies that for the class of claw-free graphs, say ${\cal G}$, $c({\cal G})$ is infinite. Joret et al.~\cite{Joret10} also showed that for  $t\geq 2$, $c($Forb$(P_t))\leq t-2$, and  later Sivaraman~\cite{Sivaraman19} showed that the robber can be captured in at most $t-1$ rounds. However, Chudnovsky et al.~\cite{Chudnovsky24} have improved the bound for $t=5$ and proved the following theorem.	
	
	\begin{theorem}[\cite{Chudnovsky24}]\label{p5freecopnum2}
		$c($\textup{Forb}$(P_5))= 2$.
	\end{theorem}
	
	Liu~\cite{Liu19} in an unpublished work has shown that if $G$ is $(P_{i_1}+P_{i_2}+\cdots +P_{i_k})$-free, then $c(G)\leq i_1+i_2+\cdots+i_k-2$, where $k\geq 2$ and $i_j\geq 3$ for some $j\in \{1,2,\ldots,k\}$. It is interesting to find $c($Forb$(p K_1))$,  $c($Forb$(qK_2))$, and $c($Forb$(pK_1+qK_2))$. Since the class of $2K_2$-free graphs is a subclass of $P_5$-free graphs, $c($Forb$(2K_2))\leq 2$. Moreover, the cop number of $C_4$ and $C_5$ is $2$. Hence $c($Forb$(2K_2))= 2$. 
	
	To proceed further, we need the following definition. A set $S\subseteq V(G)$ is a \emph{dominating set} of $G$ if every vertex in $V(G)\setminus S$ has a neighbor in $S$. The \emph{domination number}, denoted by $\gamma(G)$, is the minimum positive integer $k$ such that $G$ has a dominating set with cardinality $k$. Clearly, we have the following observation.
	
	\begin{figure}[t]
		\begin{center}
			\includegraphics[scale=0.47]{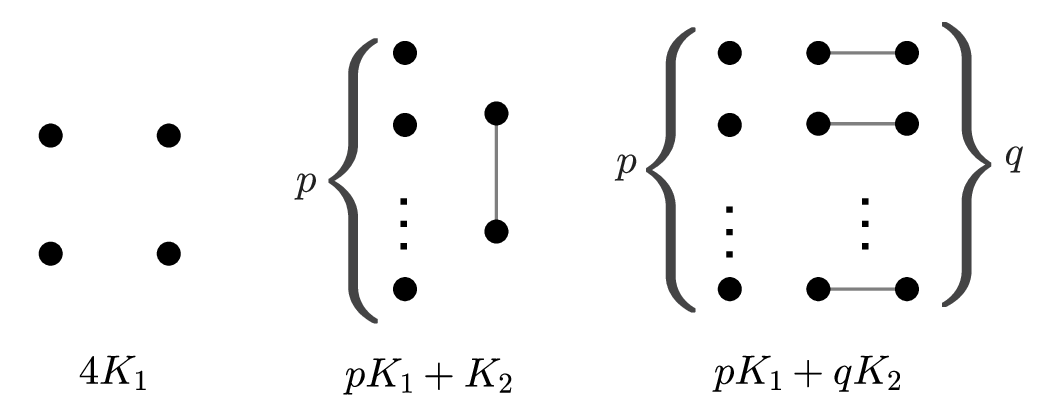}
			\caption{Some special graphs.}
			\label{Specialfig}
		\end{center}
	\end{figure}
	
	\begin{obs}\label{obs}
		For a graph $G$, $c(G)\le \gamma(G)$.
	\end{obs}   
	By Observation~\ref{obs}, we have every $pK_1$-free graph $G$ satisfies $c(G)\leq p-1$. Turcotte~\cite{Turcotte22} showed that $c($Forb$(2K_1+K_2))= 3$. In the same paper, the author has raised several questions, two of those are stated below. 
	\begin{prob}\label{pk1prob}
		Is it true that $c($\textup{Forb}$(pK_1))\leq p-2$ for $p\geq 4?$
	\end{prob}
	
	\begin{prob}\label{pk1k2prob}
		Is it true that $c($\textup{Forb}$(pK_1+K_2))\leq p+1$ for $p\geq 3?$
	\end{prob}

\subsection{Our results}	
We give an example to show that the bound mentioned in Problem~\ref{pk1prob} is not true for $p=4$ and show that there exists a $4K_1$-free graph with the cop number $3$.  In particular, we show the following theorem. 
	
\begin{theorem}\label{4k13copeg}
The complement of the Shrikhande graph is $(4K_1,C_\ell)$-free for any $\ell\geq 6$ and has the cop number $3$.
\end{theorem}

Using Theorem~\ref{4k13copeg} and Observation~\ref{obs}, we have an immediate corollary.
\begin{coro}
	$c($\textup{Forb}$(4K_1))=3$.
\end{coro}
	
Clearly, from Theorem~\ref{4k13copeg}, the complement of the Shrikhande graph is a counterexample of the following conjecture.

\begin{conj}[\cite{Sivaraman19_1}]\label{holecoj}
	If $G$ is $C_\ell$-free for all $\ell\ge 6$, then $c(G)\le 2$.
\end{conj}

Further, we  give a positive answer for Problem~\ref{pk1k2prob}. We also generalize this result for $\textup{Forb}(pK_1+qK_2)$. It is evident that for a given class of graphs, say ${\cal G}$, finding $c({\cal G})$ is nontrivial. For instance, $c($Forb$(K_{1,3}))$ is infinite. Moreover, finding $c({\cal G})$ computationally is way more challenging as it may require exhaustive search of all graphs in ${\cal G}$. Baird et al.~\cite{Baird14} showed that a graph $G$ of order $n$ with $\Delta(G)\ge n-6$ satisfies  $c(G)\leq 2$. Later Turcotte and Yvon~\cite{Turcotte21} showed that  a graph $G$ of order $n$ with $\Delta(G)\ge n-11$ satisfies $c(G)\leq 3$. In the same article, the authors motivated us to find $c({\cal G})$ computationally for a given class of graphs, say ${\cal G}$, by verifying less number of graphs using degree condition. This leads us to define the \emph{upper threshold degree} for a given class of graphs.  Let  ${\cal G}$ be a hereditary class of graphs  with $c({\cal G})\ge 2$ (for all $G\in {\cal G}$, $c(G)$ is finite). The \emph{upper threshold degree} of ${\cal G}$, denoted by $ut(\cal G)$, is defined as the positive integer $k$ such that for all $G\in {\cal G}$ having $\Delta(G)\ge |V(G)|- k$ satisfies $c(G)< c({\cal G})$ and there is a graph $G'\in {\cal G}$ having $\Delta(G')=|V(G')|- k-1$ satisfies $c(G')= c({\cal G})$. It is easy to verify that $ut($Forb$(P_5))=1$.  Now the natural question is that whether we can  reduce  the number of feasible graphs furthermore in ${\cal G}$ to obtain $c({\cal G})$ by using the minimum degree condition. This leads us to define the \emph{lower threshold degree} for ${\cal G}$.  We denote the lower threshold degree by $lt(\cal G)$.  To define the lower threshold degree, let ${\cal G'}:=\{G\in {\cal G}| c(G)=c({\cal G}) \}$.  Now we define  $lt(G):=\min \{\delta(G)|G\in {\cal G'}, |V(G)|\leq |V(G')|, \text { and } G\neq G' \text{ for all } G'\in {\cal G'} \}$.  Clearly, $lt($Forb$(P_5))=2$.  In this paper, we show the following.
	\begin{theorem}\label{upthcop}
		$ut($\textup{Forb}$(4K_1))=6$.
	\end{theorem}
	\begin{theorem}\label{lothcop}
		$lt($\textup{Forb}$(4K_1))\geq 3$.
	\end{theorem}

\subsection{Organization}
	
In Section~\ref{Shriprop}, we prove Theorem~\ref{4k13copeg}. We find the upper threshold degree of Forb$(4K_1)$ in Section~\ref{UpperThreshold}. In Section~\ref{LowerThreshold}, we find a lower bound on the lower threshold degree of Forb$(4K_1)$. We show computational results using the upper threshold degree and the lower threshold degree in Section~\ref{compute}.  Also, for $p\geq 2$ and $q\geq 2$, we prove that $c(\textup{Forb}(pK_1+K_2)) \le p+1$ and $c(\textup{Forb}(pK_1+qK_2))\le p+2q-2$ in Section~\ref{pk1qk2res}.  Finally,  we have proposed some open problems for future research in Section~\ref{future}.

\section{The cop number of the complement of the Shrikhande  graph}\label{Shriprop}
     Let $\mathbb{Z}_4=\{0,1,2,3\}$. In this section, all the arithmetic operations are in the group $A=(\mathbb{Z}_4\times \mathbb{Z}_4, +)$ unless it is mentioned otherwise. Suppose each element of $A$ is associated with a vertex and let $H$ be the graph defined on that vertex set, say $V(H)$, where the adjacency for any two tuples, say $(a,b)$ and $(c,d)$ in $V(H)$ holds if one of the following conditions is satisfied:  \\ [-20pt]
	\begin{itemize}
		\item  $a=c$ and $(b-d)\equiv  1, 3~\text{(mod } 4)$; \\ [-20pt]
		\item $(a-c)\equiv  1, 3~\text{(mod } 4)$ and $b=d$; \\ [-20pt]
		\item $(a-c)\equiv 1~\text{(mod } 4)$ and $(b-d)\equiv 1~\text{(mod } 4)$, and\\ [-20pt]
		\item $(a-c)\equiv 3~\text{(mod } 4)$ and $(b-d)\equiv 3~\text{(mod } 4)$. \\ [-20pt]
	\end{itemize}
	
		\begin{figure}[t]
			\begin{center}
			\includegraphics[width=6.6cm]{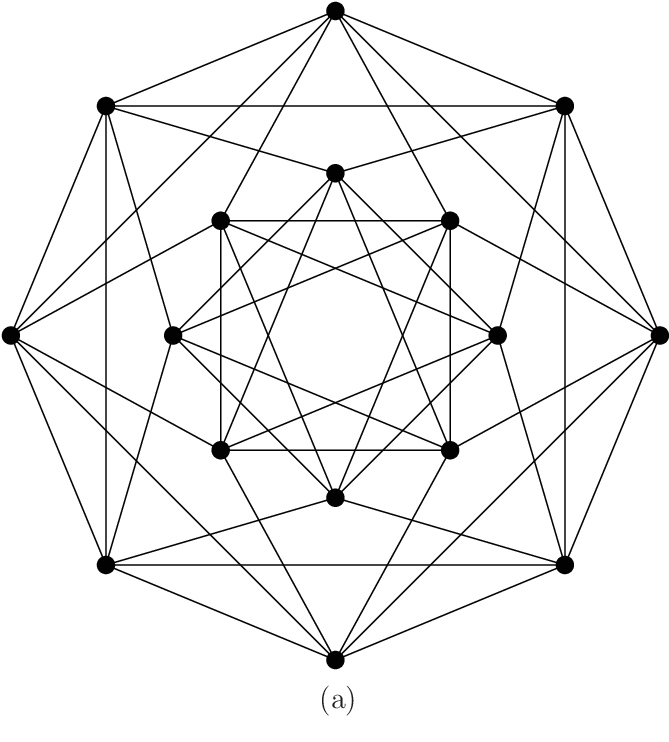}
			\hspace*{1.3cm}
			\includegraphics[width=6.6cm]{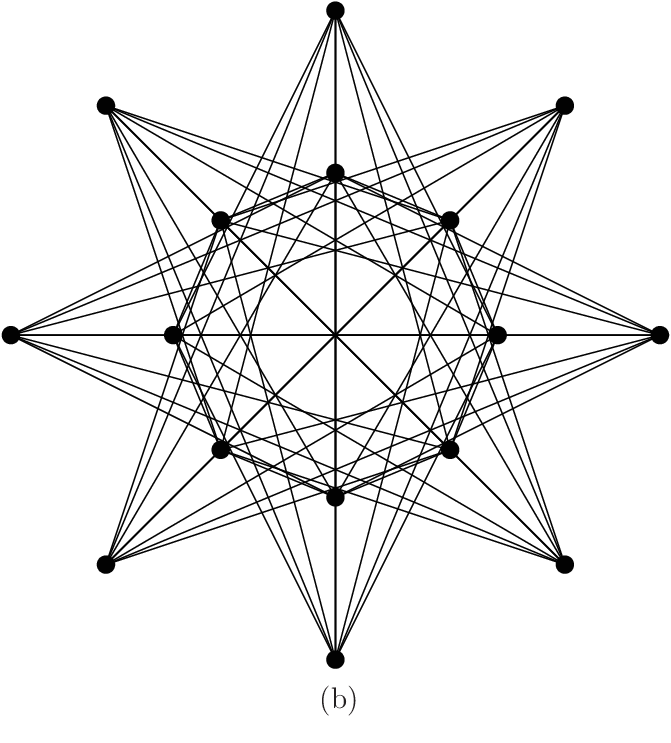}
			\caption{(a) The Shrikhande graph and (b) The complement of the Shrikhande graph\protect\footnotemark.}
			\label{ShriandShricom}
		\end{center}
		\end{figure}
		\footnotetext{Figures are generated through Sagemath.}
	
	The graph $H$ is also known as the Shrikhande graph; see Figure~\ref{ShriandShricom}(a). Although it is well known fact that  for each vertex $u\in V(H)$, $N(u)$ induces $C_6$ \cite{Peeters97} in $H$ and every pair of vertices of $H$ has exactly two common neighbors but for the completeness of our result, we will prove that $H$ is $(K_4,\overline{C_\ell})$-free, for all $\ell\ge 6$ and every pair of vertices of $H$ has exactly two common neighbors.
	\begin{lemma}\label{ShriK4}
		$H$ is $K_4$-free.
	\end{lemma}
	\begin{proof}
		Suppose $H$ contains a $K_4$ induced by the vertex set, say $K=\{x_0,x_1, x_2,x_3\}$. Since $K\subseteq V(H)$, note that each  $x_i\in K$ is of the form $(p,q)$ for $p,q\in \mathbb{Z}_4$. Next we claim the following:
		
		\begin{claim}\label{Shiraba+1bnotinK}
			If  $(a,b)\in K$, then $(a+1,b), (a,b+1), (a+3,b),(a,b+3)\notin K$. 
		\end{claim}
		\begin{proof}[Proof of Claim~\ref{Shiraba+1bnotinK}] Suppose that $(a,b)\in K$. Without loss of generality, we may assume that $x_0=(a,b)$. Then $N((a,b))=\{(a+1,b),(a+3,b), (a,b+1), (a,b+3),(a+1,b+1), (a+3,b+3)\}$. We now prove that if any of $(a+1,b), (a,b+1), (a+3,b)$, and $(a,b+3)$ belongs to $K$, then we get a contradiction. For the sake of contradiction, assume that $(a+1,b)\in K$. Without loss of generality, we may assume that $x_1=(a+1,b)$. Since $N(x_1)=N((a+1,b))=\{(a,b),(a+2,b),(a+1,b+1),(a+1,b+3),(a+2,b+1),(a,b+3)\}$, we have $N(x_0)\cap N(x_1)=\{(a+1,b+1),(a,b+3)\}$. Since $K=\{x_0,x_1,x_2,x_3\}$ induces a $K_4$, we have $x_2,x_3\in  N(x_0)\cap N(x_1)=\{(a+1,b+1),(a,b+3)\}$. Then either $x_2-x_3=(1,2)$ or $x_3-x_2=(3,2)$, a contradiction to the fact that $x_2$ is adjacent to $x_3$. This concludes the proof of Claim~\ref{Shiraba+1bnotinK}.
		\end{proof} 
        Now suppose $x_0=(a,b)$, where $(a,b)\in \mathbb{Z}_4\times \mathbb{Z}_4$. Then by Claim~\ref{Shiraba+1bnotinK}, $x_1,x_2,x_3\notin \{(a+1,b),(a,b+1),(a+3,b),(a,b+3)\}$. Then since $x_0x_1,x_0x_2,x_0x_3\in E(H)$, we have $x_1,x_2,x_3\in \{(a+1,b+1),(a+3,b+3)\}$. So $x_1,x_2$, and $x_3$ cannot be distinct contradicting to the fact that $K$ is a clique of size $4$. This concludes the proof of Lemma~\ref{ShriK4}.
\end{proof}


	
	\begin{lemma}\label{Shirstreg}
		For any two distinct vertices $u,v\in V(H)$, $|N(u)\cap N(v)|=2$. 
	\end{lemma}
	\begin{proof}
		Let $u=(a,b)$ be an arbitrary vertex of $H$. Then by the definition of the graph $H$, $N((a,b))=\{(a,b+1),(a+1,b),(a+1,b+1),(a,b+3),(a+3,b),(a+3,b+3)\}$.
		Now we have following:
		\begin{enumerate}[label=(\roman*)]
			\item Suppose that $v-u=(1,0)$. Then $v=(a+1,b)$ and thus $N((a+1,b))=\{(a+1,b+1),(a+2,b),(a+2,b+1),(a+1,b+3),(a,b),(a,b+3)\}$. Hence  $|N(u)\cap N(v)|=|\{(a+1,b+1),(a,b+3)\}|=2$. Similarly, we can conclude whenever $v-u=(3,0)$. Likewise, we can conclude whenever $v-u\in \{(0,1),(0,3)\}$.
			
			\item Next suppose that $v-u=(1,1)$. Then $v=(a+1,b+1)$ and thus $N((a+1,b+1))=\{(a+1,b+2),(a+2,b+1),(a+2,b+2),(a+1,b),(a,b+1),(a,b)\}$. Hence $|N(u)\cap N(v)|=|\{(a+1,b),(a,b+1)\}|=2$. Similarly, we can conclude whenever $v-u=(3,3)$.
			
			\item Next suppose that $v-u=(0,2)$. Then $v=(a,b+2)$. Then since $N((a,b+2))=\{(a,b+3),(a+1,b+2),(a+1,b+3),(a,b+1),(a+3,b+2),(a+3,b+1)\}$, we have $|N(u)\cap N(v)|=|\{(a,b+3),(a,b+1)\}|=2$. Likewise, we can conclude whenever $v-u=(2,0)$.
			
			\item Next suppose that $v-u=(1,2)$. Then $v=(a+1,b+2)$. Then since $N((a+1,b+2))=\{(a+1,b+3),(a+2,b+2),(a+2,b+3),(a+1,b+1),(a,b+2),(a,b+1)\}$, we have $|N(u)\cap N(v)|=|\{(a+1,b+1),(a,b+1)\}|=2$. Similarly, we can conclude whenever $v-u=(3,2)$. Likewise, we can conclude whenever $v-u\in \{(2,1),(2,3)\}$.
			
			\item Next suppose that $v-u=(1,3)$. Then $v=(a+1,b+3)$. Then since $N((a+1,b+3))=\{(a+1,b),(a+2,b+3),(a+2,b),(a+1,b+2),(a,b+3),(a,b+2)\}$, we have $|N(u)\cap N(v)|=|\{(a+1,b),(a,b+3)\}|=2$. Similarly, we can conclude whenever $v-u=(3,1)$.
			
			\item Finally, we may assume that $v-u=(2,2)$. Then $v=(a+2,b+2)$. Then since $N((a+2,b+2))=\{(a+2,b+3),(a+3,b+2),(a+3,b+3),(a+2,b+1),(a+1,b+2),(a+1,b+1)\}$, we have $|N(u)\cap N(v)|=|\{(a+3,b+3),(a+1,b+1)\}|=2$.
		\end{enumerate}
		Therefore, by (i)-(vi), for any two vertices $u,v\in V(H)$, $|N(u)\cap N(v)|=2$.
	\end{proof}
	\begin{figure}[t]
			\begin{center}
			\includegraphics[scale=.21]{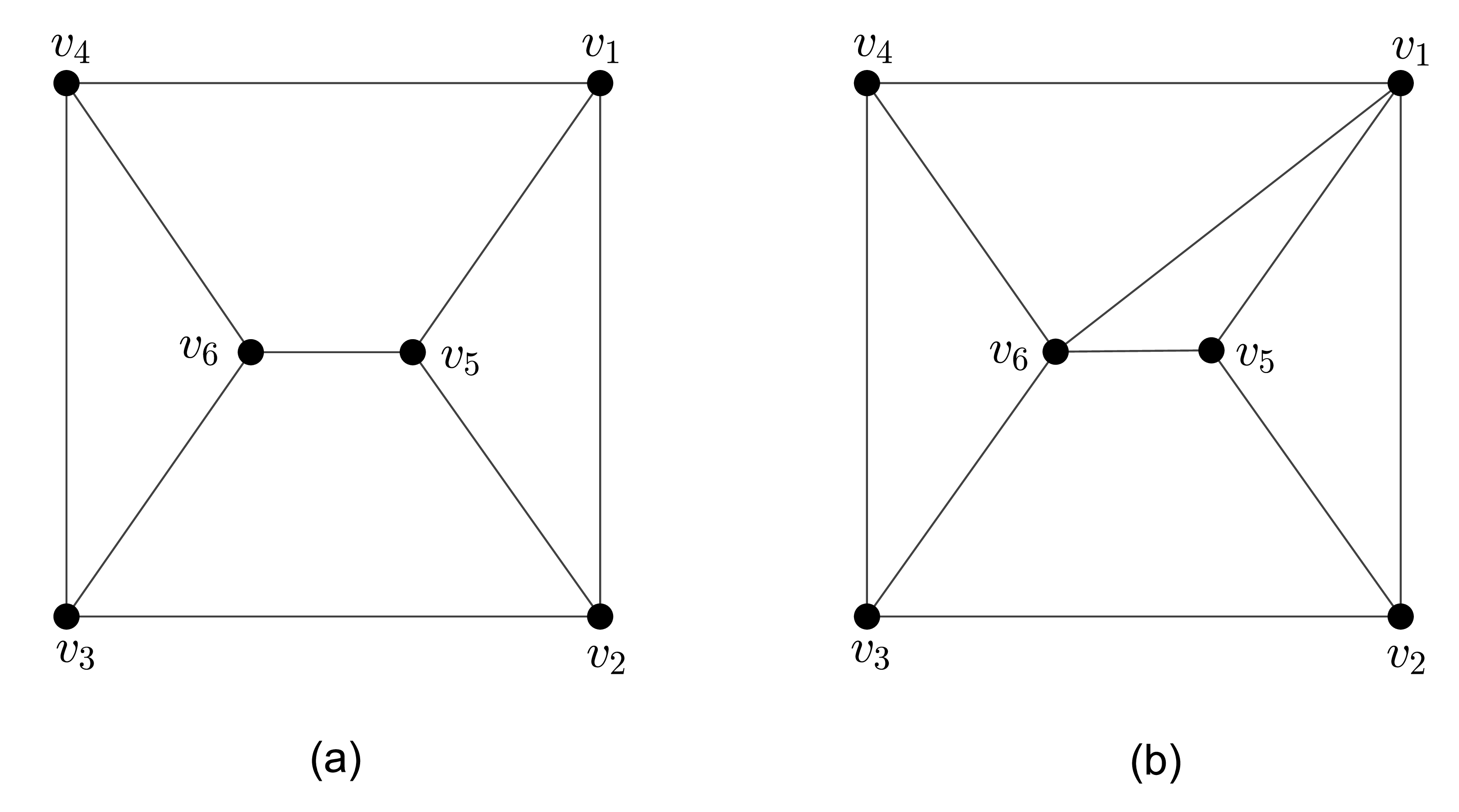}
			\caption{(a)~$\overline{C_6}$   and (b)~$\overline{P_6}$}
			\label{complmentp6c6}
				\end{center}
		\end{figure}

\begin{lemma}\label{clfree}
	$H$ is $\overline{C_\ell}$-free, for all $\ell\ge 6$. 
\end{lemma}
\begin{proof}
For the sake of contradiction, assume that $H$ contains a $\overline{C_\ell}$, where $\ell \ge 6$.  Next we claim the following.

	\begin{claim}\label{C6free}
		$\ell>6$.
	\end{claim}
	\begin{proof}[Proof of Claim~\ref{C6free}]
	For the sake of contradiction, assume that $H$ contains an induced $\overline{C_6}$ induced by the vertex set $V'=\{v_1,v_2,v_3,v_4,v_5,v_6\}$ and the edge set $E'=\{v_1v_2,v_1v_4,v_1v_5,v_2v_3,$ $v_2v_5,v_3v_4,v_3v_5,$ $v_4v_6,v_5v_6\}$ (see Figure~\ref{complmentp6c6}(a)). Since $V'\cap (N(v_5)\cap N(v_6))=\emptyset$,  by Lemma~\ref{Shirstreg},  there exist vertices in  $V(H)\setminus V'$, say $a$ and $b$ such that $N(v_5)\cap N(v_6)= \{a,b\}$.  Note that $N(v_1)\cap N(v_6)=\{v_4,v_5\}$, $N(v_2)\cap N(v_6)=\{v_3,v_5\}$, $V'\cap (N(v_1)\cap N(v_5))=\{v_2\}$, $V'\cap (N(v_2)\cap N(v_5))=\{v_1\}$. Thus by Lemma~\ref{Shirstreg}, we have $\{a,b\}$ is anticomplete to $\{v_1,v_2\}$ and there are vertices in $V(H)\setminus V'$, say $c$ and $d$ such that $c\in N(v_1)\cap N(v_5)$ and $d\in N(v_2)\cap N(v_5)$. Note that by the definition of $H$, $d(v_5)=6$. Then since $a,b,c,d,v_1,v_2,v_6\in N(v_5)$, we have $c=d$. Then clearly, $\{c,v_1,v_2,v_5\}$ induces a $K_4$, a contradiction to Lemma~\ref{ShriK4}.
	\end{proof}
    Since each induced cycle with length at least $7$ contains an induced $P_6$, by Claim~\ref{C6free},  $H$ contains an induced $\overline{P_6}$ with vertex set $V'=\{v_1,v_2,v_3,v_4,v_5,v_6\}$ and $E'=\{v_1v_2,v_1v_4,v_1v_5,v_1v_6,$ $v_2v_3,v_2v_5,v_3v_4,v_3v_6,v_4v_6,v_5v_6\}$  (see Figure~\ref{complmentp6c6}(b)). Then we have $\{v_1,v_3,v_5\}\subseteq N_H(v_2)\cap N_H(v_6)$. Therefore $|N_H(v_2)\cap N_H(v_6)|\ge 3$, a contradiction to Lemma~\ref{Shirstreg}. Hence $H$ is $\overline{C_\ell}$-free, for all $\ell\ge 6$.
\end{proof}
	
\begin{proof}[\textbf{Proof of Theorem~\ref{4k13copeg}}]
	By Lemma~\ref{ShriK4}, we have $H$ is $K_4$-free. So $\overline{H}$ is $4K_1$-free. By Lemma~\ref{clfree}, we have $\overline{H}$ is $C_\ell$-free, $\ell\ge6$. We now show that the cop number of $\overline{H}$ is $3$. For the sake of contradiction, assume that the cop number of $\overline{H}$ is $2$. At any given stage of the game, $u_i,v_i\in V(\overline{H})$ are the positions of two cops at the round $i$ of the game. If $u_i\neq v_i$, then by Lemma~\ref{Shirstreg}, there exist two distinct vertices in $V(H)=V(\overline{H})$, say $a_i$ and $b_i$, such that $a_iu_i,b_iu_i,a_iv_i,b_iv_i\in E(H)$ or in other words, $a_iu_i,b_iu_i,a_iv_i,b_iv_i\notin E(\overline{H})$. Recall that $p_i(R)$ is the position of the robber $R$ at the round $i$ of the game. We will now present a strategy for the player controlling the robber to win the game.
	
	\noindent\textbf{Case 1:} $\boldsymbol{i=1}$.
	
	If $u_i=v_i$, then we choose $p_i(R)$ to be a vertex in $V(\overline{H})$, say $p$, such that $pu_i\notin E(\overline{H})$. Such a vertex exists as $|V(H)\setminus N_{\overline{H}}[u_i]|=|N_{H}(u_i)|=6> 1$. If $u_i\neq v_i$, then we choose $p_i(R)=a_i$. Note that $a_i$ is not adjacent to $u_i$ and $v_i$.

	\noindent\textbf{Case 2:} $\boldsymbol{i>1}$.
	
	If $\{u_i,v_i\}$ is anticomplete to $p_{i-1}(R)$ in $\overline{H}$, then we choose $p_i(R)=p_{i-1}(R)$. If $\{u_i,v_i\}$ is not anticomplete to $p_{i-1}(R)$ in $\overline{H}$ and $u_i=v_i$, then we choose $p_i(R)=w$, where $w$ is a vertex in  $V(\overline{H})\setminus\{u_i\}$ such that $wp_{i-1}(R)\in E(\overline{H})$ and $u_iw\notin E(\overline{H})$. We now show that such a vertex $w$ exists. Since $\{u_i,v_i\}$ is not anticomplete to $p_{i-1}(R)$ in $\overline{H}$ and $u_i=v_i$, we have $u_ip_{i-1}(R)\in E(\overline{H})$. Now $|(N_{\overline{H}}(p_{i-1}(R))\setminus N_{\overline{H}}[u_i]|=|N_{H}(u_i)\setminus N_{H}(p_{i-1}(R)))|=|N_{H}(u_i)\setminus (N_{H}(u_i)\cap N_{H}(p_{i-1}(R)))|=|N_{H}(u_i)|-|N_{H}(u_i)\cap N_{H}(p_{i-1}(R))|$. So by Lemma~\ref{Shirstreg}, $|(N_{\overline{H}}(p_{i-1}(R))\setminus N_{\overline{H}}[u_i]|=|N_{H}(u_i)|-|N_{H}(u_i)\cap N_{H}(p_{i-1}(R))|=6-2=4>1$ and hence there exists $w\in V(\overline{H})\setminus \{u_i\}$ such that $p_{i-1}(R)w\in E(\overline{H})$ and $u_iw\notin E(\overline{H})$. 
	
If $\{u_i,v_i\}$ is not anticomplete to $p_{i-1}(R)$ in $\overline{H}$ and $u_i\neq v_i$, then we choose $p_i(R)=w'$, where $w'\in N_{\overline{H}}(p_{i-1}(R))\cap \{a_i,b_i\}$. We now show that such a vertex $w'$ exists.  If $a_ip_{i-1}(R),b_ip_{i-1}(R)\notin E(\overline{H})$, then $a_ip_{i-1}(R),b_ip_{i-1}(R)\in E(H)$;  thus $N_H(a_i)\cap N_H(b_i)=\{u_i,v_i,p_{i-1}(R)\}$, a contradiction to Lemma~\ref{Shirstreg}. So $N_{\overline{H}}(p_{i-1}(R))\cap \{a_i,b_i\}\neq \emptyset$ and let $w'\in N_{\overline{H}}(p_{i-1}(R))\cap \{a_i,b_i\}$.

By Case 1 and Case 2, at any round $i$ of the game, where $i\geq 1$, the robber can find the position $p_i(R)$ such that $u_ip_i(R), v_ip_i(R)\notin E(\overline{H})$. So we require at least $3$ cops for $\overline{H}$ to win the game. Since $\overline{H}$ is $4K_1$-free, by Observation~\ref{obs}, we have $c(\overline{H})\leq \gamma(\overline{H})\leq 3$. Thus the cop number of $\overline{H}$ is $3$. 
	\end{proof}

	\section{The upper threshold degree of Forb$(4K_1)$}\label{UpperThreshold}
	

	We use the following theorem to prove that the upper threshold degree of Forb$(4K_1)$ is $6$.
	\begin{theorem}[\cite{Baird14}]\label{cn2char}
		Let $G$ be a graph on $n$ vertices such that $G$ has a vertex, say $u$, with	degree at least $n-6$. Then either $c(G) \leq 2$ or the induced subgraph $G[V\setminus N[u]]$ is a $C_5$.
	\end{theorem}
	Let $G$ be a $4K_1$-free graph with $\Delta(G)=n-6$. So by Theorem~\ref{cn2char}, we may assume that there are vertices in $V(G)$, say $u,v_1,v_2,v_3,v_4$, and $v_5$, such that $d_G(u)=n-6$, $uv_i\notin E(G)$ for $i\in \{1,2,3,4,5\}$, and $\{v_1,v_2,v_3,v_4,v_5\}$ induces a $C_5$ with edges $v_1v_2,v_1v_5,v_2v_3,v_3v_4,$ and $v_4v_5$. Let $C:=\{v_1,v_2,v_3,v_4,v_5\}$. Let $D_1$ and $D_2$ be two cops and $R$ be the robber. In this section except the indices of $p$ and $D$, every other indices are in modulo 5. Next for $i\in \{1,2,3,4,5\}$, we define the following.
	\begin{eqnarray*}
		A_i&:=&\{v\in N(u)| N(v)\cap C=\{v_i\}\}\\
		B_i&:=&\{v\in N(u)| N(v)\cap C=\{v_{i+2},v_{i-2}\}\}\\
		X_i&:=&\{v\in N(u)| N(v)\cap C=\{v_{i+1},v_{i-1}\}\}\\
		Y_i&:=&\{v\in N(u)| N(v)\cap C=\{v_{i},v_{i+1},v_{i-1}\}\}\\
		Z_i&:=&\{v\in N(u)| N(v)\cap C=\{v_i,v_{i+2},v_{i-2}\}\}\\
		T&:=&\{v\in N(u)| N(v)\cap C=\emptyset\}\\
		L&:=&\{v\in N(u)| |N(v)\cap C|\geq 4\}
	\end{eqnarray*}
	Let $A=\bigcup\limits_{i=1}^5 A_i$, $B=\bigcup\limits_{i=1}^5B_i$, $X=\bigcup\limits_{i=1}^5X_i$, $Y=\bigcup\limits_{i=1}^5Y_i$, $Z=\bigcup\limits_{i=1}^5Z_i$. Next we have following lemmas.
	\begin{lemma}\label{prop}
		The following hold. \\ [-20pt]
		\begin{enumerate}[label=\rm{(\roman*)}]
			\item\label{unbCintatm3} $V(G)=C\cup N(u)\cup \{u\}=A\cup B \cup C\cup L\cup  X\cup Y\cup Z\cup T\cup \{u\}$. \\ [-20pt]
			
			\item \label{Xi+1Xi+1Aiclq}
			$A_i\cup A_{i+2}\cup A_{i-2}\cup B_i\cup X_{i+1}\cup X_{i-1}\cup Z_i\cup T$ is a clique. \\ [-20pt]
			
			\item\label{robpos} Suppose for some $i\in \mathbb{N}$, $p_i(D_1),p_i(D_2)\in N(u)$ and $v_j\in  N(p_i(D_1))\cup N(p_i(D_2))$ for some $j\in \{1,2,3,4,5\}$. If  $p_i(R)\in A_j\cup A_{j+1}\cup A_{j-1}\cup B_{j+2}\cup B_{j-2}\cup X_j\cup Y_j\cup T\cup \{v_j\}$, then the robber $R$ will be captured by either the cop $D_1$ or the cop $D_2$. \\ [-20pt]
		\end{enumerate}
	\end{lemma}
\begin{proof}		
(i) Since $G[V\setminus N[u]]$ induces a $C_5$, we have $C\cap N[u]=\emptyset$. Moreover, since $d_G(u)=n-6$, we have $N(u)=A\cup B \cup L\cup T\cup  X\cup Y\cup Z$. If $p\in V(G)\setminus (C\cup N[u])$, then $n=|V(G)|\geq |C|+|N[u]|+1=5+n-6+1+1=n+1$, a contradiction. So we have $V(G)=C\cup N[u]$; thus $V(G)=A\cup B \cup C\cup L\cup  X\cup Y\cup Z\cup T\cup \{u\}$. This proves \ref{unbCintatm3}. 
			
(ii) If there are nonadjacent vertices, say $a,b\in A_i\cup A_{i+2}\cup A_{i-2}\cup B_i\cup X_{i+1}\cup X_{i-1}\cup Z_i\cup T$, then $\{a,b,v_{i+1},v_{i-1}\}$ induces a $4K_1$, a contradiction. So $A_i\cup A_{i+2}\cup A_{i-2}\cup B_i\cup X_{i+1}\cup X_{i-1}\cup Z_i\cup T$ is a clique. This proves \ref{Xi+1Xi+1Aiclq}. 
			
(iii) Suppose $p_i(R)\in A_j\cup A_{j+1}\cup A_{j-1}\cup B_{j+2}\cup B_{j-2}\cup X_j\cup Y_j\cup T \cup \{v_j\}$. We may assume that $v_j\in N(p_i(D_1))$. Next $p_{i+1}(D_1)=v_j$ and $p_{i+1}(D_2)=u$. Note that $N(u)=A\cup B \cup L\cup T\cup  X\cup Y\cup Z$. This implies that $N[p_i(R)]\subseteq N[u]\cup N[v_j]$ and hence $p_{i+1}(R)\in N[u]\cup N[v_j]$. So the robber $R$ will be captured either the cop $D_1$ or the cop $D_2$. This proves \ref{robpos}. 
\end{proof}
	
	\setcounter{claim}{0}
	\begin{lemma}\label{LZemp}
		If $L\cup Z\neq \emptyset$, then $c(G)\leq 2$.
	\end{lemma}
	\begin{proof}
		First we claim the following:
		\begin{claim}\label{Lemp}
			We may assume that $L=\emptyset$.
		\end{claim}
		\begin{proof}[Proof of Claim~\ref{Lemp}]
		For the sake of contradiction, assume that $L\neq \emptyset$. Then there is a vertex in $N(u)$, say $p$, such that $p\in L$ and $4\leq |C\cap N(p)|\leq 5$. If $|C\cap N(p)|=5$, then since $\{p\}$ is complete to $C$, $\{u,p\}$ is a dominating set of $G$. Moreover, by Observation~\ref{obs}, $c(G)\leq 2$. So we assume $|C\cap N(p)|=4$. Without loss of generality, assume that $pv_1\notin E(G)$. Now we will show two cops are enough to catch the robber. Now for all $i\in \mathbb{N}$, $p_i(D_1)=u$. So $p_1(R)\in C$ and by Lemma~\ref{prop}(i), for all $i\in \mathbb{N}$ and $i\geq 2$, $p_i(R)\in C\cap N(p_{i-1}(R))$. Now $p_1(D_2)=p$. Then $p_1(R)=v_1$. Next $p_2(D_2)=v_5$. This implies that $p_2(R)=v_2$. Next $p_3(D_2)=v_1$ and thus $p_3(R)=v_3$. After that $p_4(D_2)=v_2$ and so $p_4(R)=v_4$. Next $p_5(D_2)=p$. Now since $N[v_4]\cap C\subset N(p)$, the robber $R$ will be captured in the round $5$ of the game. So we may assume that $L=\emptyset$.
		\end{proof}
		So by Claim~\ref{Lemp}, $Z\neq \emptyset$. For $k\in \{1,2,3,4,5\}$, let $I_k:=\{l\in \{1,2,3,4,5\}|N(v_k)\cap Z_l\neq \emptyset\}$. Note that $I_k\subseteq\{k-2,k,k+2\}$ and  $|I_k|\leq 3$. Next we claim the following.
		
		\begin{claim}\label{Ikcard2}
			We may assume that for all $k\in \{1,2,3,4,5\}$, if $|I_k|=2$, then for any $a,b\in I_k$, $|a-b|\equiv\pm1\Mod{5}$.
		\end{claim}
		\begin{proof}[Proof of Claim~\ref{Ikcard2}]
		For the sake of contradiction, assume that there exists a set $I_k$ for some $k\in \{1,2,3,4,5\}$ such that $|I_k|=2$ and for any $a,b\in I_k$, $|a-b|\not\equiv \pm1\pmod{5}$. Without loss of generality, assume that $|I_1|=2$ and so either $I_1=\{1,3\}$ or $I_1=\{1,4\}$. By symmetry, we may assume that $I_1=\{1,3\}$. So $Z_4=\emptyset$. Let $z_1\in Z_1$ and $z_3\in Z_3$. Now $p_1(D_1)=u$ and $p_1(D_2)=v_1$. So by Lemma~\ref{prop}(i), we have $p_1(R)\in \{v_3,v_4\}$. First suppose $p_1(R)=v_3$. Next we have $p_2(D_1)=u$ and $p_2(D_2)=v_2$ and so we have $p_2(R)=v_4$. Now suppose $Z_2\cup Z_5\neq \emptyset$. Let $z\in Z_2\cup Z_5$. 
		Then $p_3(D_1)=z_1$ and $p_3(D_2)=z$. Then $p_3(R)\notin C$. Also by Lemma~\ref{prop}(iii), we may assume that $p_3(R)\notin A\cup B\cup X\cup Y$ and by Lemma~\ref{prop}(ii), each $Z_i$ is a clique; so we may assume that $p_3(R)\notin Z_1\cup Z_2$. So $p_3(R)\in Z_4$, a contradiction to the fact $Z_4=\emptyset$. Hence we assume that $Z_2\cup Z_5=\emptyset$. Then $p_3(D_1)=u$ and $p_3(D_2)=v_3$ which implies that $p_3(R)=v_5$. Now $p_4(D_1)=z_1$ and $p_4(D_2)=z_3$. Since $\{v_5\}$ is complete to $A_5\cup B_2\cup B_3\cup X_1\cup X_4\cup Y_1\cup Y_4\cup Y_5\cup Z_3\cup \{v_1,v_4\}$, by Lemma~\ref{prop}(ii) and (iii), $c(G)\leq 2$. 
		
		So we assume that $p_1(R)=v_4$. Now if $Z_2\cup Z_5\neq \emptyset$, then for any $z\in Z_2\cup Z_5$, we have $p_2(D_1)=z_1$ and $p_2(D_2)=z$, and by earlier argument, we have $c(G)\leq 2$. So we assume that $Z_2\cup Z_5=\emptyset$. Then $p_2(D_1)=u$ and $p_2(D_2)=z_1$ which implies that $p_2(R)=v_5$. Now $p_3(D_1)=z_3$ and $p_3(D_2)=z_1$ and also by similar argument to earlier one, we have $c(G)\leq 2$. So  we may assume that for all $k\in \{1,2,3,4,5\}$, if $|I_k|=2$, then for any $a,b\in I_k$, $|a-b|\equiv\pm1\Mod{5}$.
		\end{proof}
		
		\begin{claim}\label{Ikcard2case}
			We may assume that for all $k\in \{1,2,3,4,5\}$, $|I_k|\neq 2$.
		\end{claim}
		\begin{proof}[Proof of Claim~\ref{Ikcard2case}]
		For the sake of contradiction, assume that there is a $k\in \{1,2,3,4,5\}$ such that $|I_k|=2$. Without loss of generality, we may assume that $k=1$. Then by Claim~\ref{Ikcard2}, $Z_3$ and $Z_4$ are nonempty and $Z_1$ is empty. Since $Z_1=\emptyset$ and $Z_3\neq \emptyset$, by Claim~\ref{Ikcard2}, we have $Z_5=\emptyset$. Similarly, $Z_2= \emptyset$. Let $z_3\in Z_3$ and $z_4\in Z_4$. Now $p_1(D_1)=v_1$ and $p_1(D_2)=u$. Then $p_1(R)\in \{v_3,v_4\}$. Next $p_2(D_1)=z_3$ and $p_2(D_2)=z_4$. Since $(N(z_3)\cup N(z_4))\cap C=C$, by Lemma~\ref{prop}(iii), we have $c(G)\leq 2$. So we assume that for all  $k\in \{1,2,3,4,5\}$, $|I_k|\neq 2$. 
	    \end{proof}
	    
		\begin{claim}\label{Ikcard1}
			We may assume that for all $k\in \{1,2,3,4,5\}$,  $|I_k|= 3$.
		\end{claim}
		\begin{proof}[Proof of Claim~\ref{Ikcard1}]
		For the sake of contradiction, assume that there is a $k\in \{1,2,3,4,5\}$ such that $|I_k|\neq 3$. By Claim~\ref{Lemp}, $Z\neq \emptyset$ and thus without loss of generality, we may assume that $Z_1\neq \emptyset$. If $Z_3\cup Z_4\neq \emptyset$, then since $|I_1|\neq 2$, by Claim~\ref{Ikcard2case}, we have $Z_3$ and $Z_4$ are nonempty and thus since $|I_3|\neq 2$ and $|I_4|\neq 2$, we have $Z_2$ and $Z_5$ are nonempty; so for all $k\in \{1,2,3,4,5\}$, $|I_k|= 3$, a contradiction. So we assume that $Z_3\cup Z_4=\emptyset$. Since $|I_3|\neq 2$, by Claim~\ref{Ikcard2case}, we have $Z_5=\emptyset$. Similarly, we have $Z_2=\emptyset$.
		Let $z_1\in Z_1$. Now $p_1(D_1)=z_1$ and $p_1(D_2)=u$ and so $p_1(R)\in \{v_2,v_5\}$. By symmetry, we may assume that $p_1(R)=v_2$. Suppose $N(v_2)\cap (A\cup B\cup X\cup Y)\neq \emptyset$. Let $p\in N(v_2)\cap (A\cup B\cup X\cup Y)$. Then $p_2(D_1)=z_1$ and $p_2(D_2)=p$. Thus by Lemma~\ref{prop}(iii), we have $c(G)\leq 2$. So we assume that $N(v_2)\cap (A\cup B\cup X\cup Y)= \emptyset$. Now $p_2(D_1)=v_3$ and $p_2(D_2)=z_1$ and thus  $p_2(R)\in \{v_1,v_2,v_3\}$; so the robber $R$ will be captured in the round $2$ of the game. Hence we assume that $p_1(R)=v_5$ and the rest of the proof is similar to earlier one. So we may assume that for all $k\in \{1,2,3,4,5\}$, $|I_k|= 3$.
		\end{proof}
		So by Claim~\ref{Ikcard1}, we have for all $k\in \{1,2,3,4,5\}$, we have $Z_k\neq \emptyset$. Next we have following.
		
		\begin{claim}\label{Z1Z3Z4ancom}
			We may assume that for all $k\in \{1,2,3,4,5\}$, $Z_k$ is anticomplete to $Z_{k+2}\cup Z_{k-2}$.
		\end{claim}
		\begin{proof}[Proof of Claim~\ref{Z1Z3Z4ancom}]
		For the sake of contradiction, assume that that there is a $k\in \{1,2,3,4,5\}$ such that $Z_k$ is not anticomplete to $Z_{k+2}\cup Z_{k-2}$. Without loss of generality, we may assume that $Z_1$ is not anticomplete to $Z_{3}\cup Z_{4}$. Then there are adjacent vertices, $p\in Z_1$ and $q\in Z_3\cup Z_4$. We start with $p_1(D_1)=v_3$ and $p_1(D_2)=u$. So $p_1(R)\in \{v_1,v_5\}$. Suppose $p_1(R)=v_1$. First suppose $q\in Z_3$. Then $p_2(D_1)=q$ and $p_2(D_2)=v$, where $v\in Z_4$. So by Lemma~\ref{prop}(iii), $p_2(R)\notin A\cup B\cup C \cup X\cup Y$ and also by Lemma~\ref{prop}(ii), $p_2(R)\notin Z_3\cup Z_4$. Thus we have $p_2(R)\in Z_1$. Then $p_3(D_1)=p$ and $p_3(D_2)=u$ which implies that $p_3(R)\in C\cap N[p_2(R)]=\{v_1,v_3,v_4\}$. 
		Thus the robber $R$ will be captured in the round $3$ of the game. Similarly, when $q\in Z_4$, then $p_2(D_1)=v'$, where $v'\in Z_3$ and $p_2(D_2)=q$ and the rest of the proof is similar to the earlier case.  So we assume that $p_1(R)=v_5$. Then $p_2(D_1)=v_4$ and $p_2(D_2)=u$ which implies that $p_2(R)=v_1$. Now rest of the proof is similar to the earlier case. So we have $Z_1$ is anticomplete to $Z_3\cup Z_4$.
		\end{proof}
		
		Let $z_k'\in Z_k$, where $k\in \{1,2,3,4,5\}$. Next we have following.
		\begin{claim}\label{Ziindices}
			We may assume that for all $k\in \{1,2,3,4,5\}$, $\{z_k'\}$ is not complete to $\{z_{k+1}',z_{k-1}'\}$.
		\end{claim}
		\begin{proof}[Proof of Claim~\ref{Ziindices}]
		For the sake of contradiction, assume that there is a $k\in \{1,2,3,4,5\}$ such that $\{z_k'\}$ is complete to $\{z_{k+1}',z_{k-1}'\}$. Without loss of generality, assume that $k=1$. Then we start with $p_1(D_1)=z_1'$ and $p_1(D_2)=u$. Now $p_1(R)\in \{v_2,v_5\}$. By symmetry, we may assume that $p_1(R)=v_2$. Then $p_2(D_1)=z_1'$ and $p_2(D_2)=z_4'$. Then by Lemma~\ref{prop}(iii), we have $p_2(R)\in Z_2\cup Z_5$. Next $p_3(D_2)=u$ which implies that $p_3(R)\in C\cap N(p_2(R))$. If $p_2(R)\in Z_2$, then $p_3(D_1)=z_2'$ and if  $p_2(R)\in Z_5$, then $p_3(D_1)=z_5'$. So the robber $R$ will be captured by $D_1$ in the round $3$ of the game. 
		\end{proof}
		
		\begin{claim}\label{Zifinalind}
			There are indices $k,\ell$ such that $k,\ell\in \{1,2,3,4,5\}$ and $\ell\notin \{k,k+1\}$, we have $z_k'z_{k+1}',z_\ell'z_{\ell+1}'\in E(G)$. 
		\end{claim}
		\begin{proof}[Proof of Claim~\ref{Zifinalind}]
		By Claim~\ref{Ziindices}, without loss of generality, we may assume that $z_1'z_2'\notin E(G)$. Also by Claim~\ref{Z1Z3Z4ancom}, we have $z_1'z_3',z_1'z_4',z_2'z_4'\notin E(G)$. Since $\{z_1',z_2',z_3',z_4'\}$ does not induce a $4K_1$, we have either $z_2'z_3'\in E(G)$ or $z_3'z_4'\in E(G)$. First we assume that $z_2'z_3'\in E(G)$. Then by Claim~\ref{Ziindices}, $z_3'z_4'\notin E(G)$. If $z_4'z_5'\in E(G)$, then $k=2$ and $\ell=4$ will be our desired index, so we assume that $z_4'z_5'\notin E(G)$. If $z_1'z_5'\in E(G)$, then $k=5$ and $\ell=2$ will be our desired index. So we may assume that $z_1'z_5'\notin E(G)$.  Then by Claim~\ref{Z1Z3Z4ancom}, $\{z_1',z_3',z_4',z_5'\}$ induces a $4K_1$, a contradiction. 
		
		Finally, we assume that $z_2'z_3'\notin E(G)$ and $z_3'z_4'\in E(G)$. Since by Claim~\ref{Z1Z3Z4ancom}, $\{z_1',z_2',z_3',z_5'\}$ does not induce a $4K_1$, we have $z_1'z_5'\in E(G)$. Then $k=3$ and $\ell=5$ will be our desired indices.
		\end{proof}

		\medskip
		By Claim~\ref{Zifinalind}, without loss of generality, assume that $z_1'z_2', z_3'z_4'\in E(G)$.  Now we start with $p_1(D_1)=v_2$ and $p_1(D_2)=u$. Then $p_1(R)\in \{v_4,v_5\}$. Suppose $p_1(R)=v_4$. Then $p_2(D_1)=z_2'$ and $p_2(D_2)=z_3'$. Then by Lemma~\ref{prop}(iii), we have $p_2(R)\in Z_1\cup Z_4$. If $p_2(R)\in Z_1$, then $p_3(D_1)=z_1'$ and $p_3(D_2)=u$ and thus $p_3(R)\in \{v_1,v_3,v_4\}$; so the robber $R$ will be captured by the cop $D_1$ in the round $3$ of the game. If $p_2(R)\in Z_4$, then $p_3(D_1)=u$ and $p_3(D_2)=z_4'$ and thus $p_3(R)\in \{v_1,v_2,v_4\}$; so the robber $R$ will be captured the cop $D_2$ in the round $3$ of the game. So we assume that $p_1(R)=v_5$. Then $p_2(D_1)=v_1$ and $p_2(D_2)=u$. Then we have $p_2(R)=v_4$. Now $p_2(D_1)=z_3'$, $p_2(D_2)=z_2'$. Then by Lemma~\ref{prop}(iii), we have $p_2(R)\in Z_1\cup Z_4$ and the rest of the proof is similar when $p_1(R)=v_4$. So $c(G)\leq 2$.  This completes the proof of Lemma~\ref{LZemp}.  
	\end{proof}
	\setcounter{claim}{0}
	\begin{lemma}\label{Yempcop}
		If $Y\neq \emptyset$, then $c(G)\leq 2$.
	\end{lemma}
	\begin{proof}
		By Lemma~\ref{LZemp}, we have $L\cup Z=\emptyset$.
		Next we claim the following:
		\begin{claim}\label{YiYi+2Yi-2oneemp}
			We may assume that for all $i\in \{1,2,3,4,5\}$, at least one of $Y_i$ and $Y_{i+2}\cup Y_{i-2}$ is empty.
		\end{claim}
		\begin{proof}[Proof of Claim~\ref{YiYi+2Yi-2oneemp}]
		For the sake of the contradiction, assume that there is a $i\in\{1,2,3,4,5\}$ such that both $Y_i$ and $Y_{i+2}\cup Y_{i-2}$ are nonempty. Without loss of generality, assume that $i=1$. Then $Y_1$ and $Y_3\cup Y_4$ are nonempty. Note that by Lemma~\ref{prop}(i), if there is a cop in $u$, $p_j(R)\in C$. Let $y_1\in Y_1$ and $y^*\in Y_3\cup Y_4$. By symmetry, we may assume $y^*\in Y_3$. Now $p_1(D_1)=y_1$, $p_1(D_2)=u$. Then since $N(y_1)=\{v_5,v_1,v_2\}$, we have $p_1(R)\in \{v_3,v_4\}$. First suppose $p_1(R)=v_3$. Then  $p_2(D_1)=u$ and  $p_2(D_2)=y^*$. Now since $p_1(R)\notin N(u)$, we have $p_2(R)\in \{v_2,v_3,v_4\}$ and thus the robber $R$ will be captured in the round $2$ of the game. So $c(G)\leq 2$. So we have $p_1(R)=v_4$. Then $p_2(D_1)=v_5$ and $p_2(D_2)=u$. So $p_2(R)=v_3$. Next $p_3(D_1)=v_4$ and $p_3(D_2)=u$, and thus $p_3(R)=v_2$. Next $p_4(D_1)=y^*$ and $p_4(D_2)=u$, and thus $p_4(R)=v_1$. Next $p_5(D_1)=u$ and $p_5(D_2)=y_1$. Thus $p_5(R)\in \{v_1,v_2,v_5\}\subseteq N(y_1)$. So the robber $R$ will be captured in the round $5$ by $D_2$. Hence $c(G)\leq 2$. So we may assume that at least one of $Y_i$ and $Y_{i+2}\cup Y_{i-2}$ is empty.
		\end{proof}
		By Claim~\ref{YiYi+2Yi-2oneemp}, we may assume that $Y_1\cup Y_2\neq \emptyset$ and $Y_3\cup Y_4\cup Y_5=\emptyset$. Now we split the proof into two cases.

	\noindent\textbf{Case 1:} Either $Y_1=\emptyset$  or $Y_2=\emptyset$.
	
		We may assume that $Y_2$ is empty. Let $y_1\in Y_1$. We start with $p_1(D_1)=y_1$ and $p_1(D_2)=u$. So by Lemma~\ref{prop}(i), we have $p_1(R)\in \{v_3,v_4\}$. Next we claim the following.
		\begin{claim}\label{Nv3Nv4emp}
			We may assume that either $N(v_3)\cap N(u)= \emptyset$ or $N(v_4)\cap N(u)=\emptyset$.
		\end{claim}
		\begin{proof}[Proof of Claim~\ref{Nv3Nv4emp}]
		For the sake of contradiction, assume that $N(v_3)\cap N(u)\neq \emptyset$ and $N(v_4)\cap N(u)\neq \emptyset$. Let $v\in N(v_3)\cap N(u)$ and $w\in N(v_4)\cap N(u)$. Suppose $p_1(R)=v_3$. Next $p_2(D_1)=y_1$ and $p_2(D_2)=v$. Then by Lemma~\ref{prop}(i) and (iii), we have $p_2(R)=v_4$. Next $p_3(D_1)=u$ and $p_3(D_2)=v_3$ which implies that $p_3(R)=v_5$. Next $p_4(D_1)=u$ and $p_4(D_2)=v_4$ which implies that $p_4(R)=v_1$. Now $p_5(D_1)=y_1$ and $p_5(D_2)=w$. Then $p_5(R)\in A_1\cup B_3\cup B_4\cup X_2\cup X_5\cup Y_1$. Then by Lemma~\ref{prop}(iii), we have $c(G)\leq 2$. Similarly, when $p_1(R)=v_4$, we have $c(G)\leq 2$. So we assume that either $N(v_3)\cap N(u)=\emptyset$ or $N(v_3)\cap N(u)= \emptyset$.
		\end{proof}

		\begin{claim}\label{Nv3emp}
			We may assume that $N(v_3)\cap N(u)= \emptyset$. Likewise $N(v_4)\cap N(u)= \emptyset$.
		\end{claim}
		\begin{proof}[Proof of Claim~\ref{Nv3emp}]
		For the sake of contradiction, assume that $N(v_3)\cap N(u)\neq \emptyset$. Let $v\in N(v_3)\cap N(u)$. Then by Claim~\ref{Nv3Nv4emp}, $N(v_4)\cap N(u)=\emptyset$. Suppose $p_1(R)=v_3$. Next $p_2(D_1)=y_1$ and $p_2(D_2)=v$. Then by Lemma~\ref{prop}(iii), we have $p_2(R)=v_4$. Next $p_3(D_1)=y_1$ and $p_3(D_2)=v_3$ which implies that $p_3(R)\in \{v_3,v_4,v_5\}$. So the robber $R$ will be captured in the round $3$ of the game. Now suppose $p_1(R)=v_4$. Then $p_2(D_1)=y_1$ and $p_2(D_2)=v$. Then by Lemma~\ref{prop}(iii), we have $p_2(R)=v_4$ and by earlier argument, we conclude that $c(G)\le 2$.
		\end{proof}
		Now by symmetry, we may assume that $p_1(R)=v_3$. Next $p_2(D_1)=v_2$ and $p_2(D_1)=y_1$. Since Claim~\ref{Nv3emp}, $N(v_3)\cap N(u)= \emptyset$, we have $p_2(R)=v_4$. Next $p_3(D_1)=v_3$ and $p_3(D_2)=y_1$ and again by Claim~\ref{Nv3emp}, we have $p_3(R)\in \{v_3,v_4,v_5\}$. Hence the robber $R$ will be captured in the round $3$ of the game. Hence $c(G)\leq 2$.

	\noindent\textbf{Case 2:} $Y_1\neq\emptyset$ and $Y_2\neq\emptyset$.
	
		Let $y_1\in Y_1$ and $y_2\in Y_2$. We start with $p_1(D_1)=y_1$ and $p_1(D_2)=y_2$. By Lemma~\ref{prop}(iii), we have $p_1(R)\notin A\cup B\cup X_1\cup X_2\cup X_3\cup X_5$. So $p_1(R)\in X_4\cup \{v_4\}$. Next we have following:
		\begin{claim}\label{I2X4emp}
			We may assume that $X_4=\emptyset$.
		\end{claim}
		\begin{proof}[Proof of Claim~\ref{I2X4emp}]
		For the sake of contradiction, assume that $X_4\neq \emptyset$. Let $x\in X_4$. Now $p_2(D_1)=v_5$ and $p_2(D_2)=u$, and thus by Lemma~\ref{prop}(i), we have $p_2(R)=v_3$. Next $p_{3}(D_1)=v_5$ and $p_{3}(D_2)=y_2$. Then $p_{3}(R)\in A_3\cup B_1\cup B_5\cup X_2$. If $p_{3}(R)\in A_3\cup B_1$, then $p_{4}(D_1)=v_4$ and $p_{4}(D_2)=u$. Therefore $p_{4}(R)\subseteq N[v_4]$ and the robber $R$ will be captured in the round $4$ of the game. If $p_{3}(R)\in X_4$, then $p_{4}(D_1)=x$ and $p_{4}(D_2)=u$. By Lemma~\ref{prop}(ii), $p_{4}(R)\in\{v_3,v_5\}\subseteq N[x]$. So the robber $R$ will be captured in the round $4$ of the game. Finally, $p_{3}(R)\in B_5\cup X_2$. Then $p_{4}(D_1)=x$ and $p_{4}(D_2)=u$ which implies that $p_{4}(R)\in \{v_1,v_2\}$. Now $p_{5}(D_1)=u$. If $p_{4}(R)=v_1$, then $p_{5}(D_2)=y_1$ and if $p_{4}(R)=v_2$, then $p_{5}(D_2)=y_2$. Since $N(v_1)\cap C\subseteq N(y_1)$ and $N(v_2)\cap C\subseteq N(y_2)$, the robber $R$ will be captured in the round $5$ of the game. So we assume that $X_4=\emptyset$. 
		\end{proof}
		By Claim~\ref{I2X4emp}, we assume that $X_4=\emptyset$. So $p_1(R)=v_4$. Now $p_2(D_1)=v_5$ and $p_2(D_2)=u$ which implies that $p_2(R)=v_3$. Next we have $p_3(D_1)=v_4$ and $p_3(D_2)=u$ which implies that $p_3(R)=v_2$. Next  $p_4(D_1)=v_3$ and $p_3(D_2)=u$ which implies that $p_4(R)=v_1$. Now $p_5(D_1)=y_2$ and $p_5(D_2)=y_1$ which implies that $p_5(R)\in A_1\cup B_3\cup B_4\cup X_2\cup X_5\cup Y_1\cup Y_2$. So by Lemma~\ref{prop}(iii), we have $c(G)\leq 2$. 
	\end{proof}
	\setcounter{claim}{0}
\begin{lemma}\label{upthcomp}
			Let $G$ be a $4K_1$-free graph with $\Delta(G)=n-6$. Then $c(G)\leq 2$.
\end{lemma}
	\begin{proof}
	By Lemma~\ref{LZemp} and by Lemma~\ref{Yempcop}, $L\cup Y\cup Z=\emptyset$. To proceed further we claim the following.
		\begin{claim}\label{Aclq}
			$A$ is a clique.
		\end{claim}
	\begin{proof}[Proof of Claim~\ref{Aclq}]
	 For the sake of contradiction, assume that there are nonadjacent vertices in $A$, say $a$ and $b$. By Lemma~\ref{prop}(ii), we may assume that $a\in A_i$ and $b\in A_{i+1}\cup A_{i-1}$. By symmetry, we may assume that $b\in A_{i+1}$. Then $\{a,b,v_{i+2},v_{i-1}\}$ induces a $4K_1$, a contradiction. Thus $A$ is a clique. This proves Claim~\ref{Aclq}.
	 \end{proof}
	 
	\begin{claim}\label{unbXnemp}
			We may assume that $B\cup X\neq \emptyset$.	
	\end{claim}
	\begin{proof}[Proof of Claim~\ref{unbXnemp}]
	For the sake of contradiction, assume that $B\cup X=\emptyset$. Since $G$ is connected, we have $A\neq \emptyset$. Let $a\in A$. We may assume that $a\in A_1$. For all $i\in \mathbb{N}$, $p_i(D_1)=a$. Thus by Lemma~\ref{prop}(ii) and Claim~\ref{Aclq}, the  $p_i(R)\in C\setminus \{v_1\}$. Now $p_1(D_2)=v_3$ which implies $p_1(R)=v_5$. Next $p_2(D_2)=v_4$, then since $N(v_5)\subset N(a)\cup N(v_4)\cup \{v_4\}$, the robber $R$ will be captured in the round $2$ of the game. Thus $c(G)\leq 2$. So we may assume that $B\cup X\neq \emptyset$. This proves Claim~\ref{unbXnemp}.
	\end{proof}
		
	\begin{claim}\label{unbXiXi+1nemp}
		We may assume that for all $i\in \{1,2,3,4,5\}$, either $X_i=\emptyset$ or $X_{i+1}=\emptyset$.  	
	\end{claim}
	\begin{proof}[Proof of Claim~\ref{unbXiXi+1nemp}]
	For the sake of contradiction, assume that there is an $i\in\{1,2,3,4,5\}$ such that $X_i\neq\emptyset$ and $X_{i+1}\neq \emptyset$. Without loss of generality, we may assume that $i=1$. Then $X_1$ and $X_2$ are nonempty. Let $x_1\in X_1$ and $x_2\in X_2$. Let $p_1(D_1)=x_1$ and $p_1(D_2)=x_2$. Since by Lemma~\ref{prop}(ii), $\{x_1\}$ is complete to $X_3\cup X_4$ and $\{x_2\}$ is complete to $X_4\cup X_5$, by Lemma~\ref{prop}(iii), we have $p_1(R)=v_4$. Next $p_2(D_1)=v_5$ and $p_2(D_2)=x_2$. Now $p_2(R)\in X_3$; otherwise $p_2(R)\in N(v_4)\setminus X_3\subseteq N[x_2]\cup N[v_5]$, and thus the robber $R$ will be captured in the round $2$ of the game. So $p_2(R)=x$, for some $x\in X_3$. Next $p_3(D_1)=x_1$. 
		
	Now suppose $(N(v_4)\cap N(u))\setminus (X_3\cup B_2)\neq \emptyset$.   Let $x'\in (N(v_4)\cap N(u))\setminus (X_3\cup B_2)$=$A_4\cup B_1\cup X_5$. By Lemma~\ref{prop}(ii), $x_2x'\in E(G)$. We choose  $p_3(D_2)=x'$. Since $p_3(R)\in A\cup B\cup X\cup \{v_2,v_4\}$, by Lemma~\ref{prop}(ii) and (iii), we have $c(G)\leq 2$.  
		
	So we assume that  $(N(v_4)\cap N(u))\setminus (X_3\cup B_2)= \emptyset$. Then $p_3(D_2)=x_2$. Again since $p_3(R)\in A\cup B\cup X\cup \{v_2,v_4\}$, by Lemma~\ref{prop}(ii) and (iii), we have $p_3(R)= v_4$.  Next $p_4(D_1)=x_1$ and $p_4(D_2)=v_3$. Since by Lemma~\ref{prop}(ii), $\{x_1\}$ is complete to $X_3\cup B_2$, the robber $R$ will be captured in the round $4$ of the game. So we may assume that either $X_i=\emptyset$ or $X_{i+1}=\emptyset$.
	\end{proof}

		\begin{claim}\label{unbXemp}
			We may assume that  $X=\emptyset$.  	
		\end{claim}
	\begin{proof}[Proof of Claim~\ref{unbXemp}]
	For the sake of contradiction, assume that $X\neq \emptyset$. So by Claim~\ref{unbXiXi+1nemp}, we may assume that $X_2\cup X_5\neq \emptyset$ and $X_1\cup X_3\cup X_4=\emptyset$. Suppose $X_2\neq \emptyset$. Let $x_2\in X_2$. Now we start with, $p_1(D_1)=v_5$ and $p_1(D_2)=u$. Then $p_1(R)\in \{v_2,v_3\}$. Suppose $p_1(R)=v_2$. Then $p_2(D_1)=v_1$ and $p_2(D_2)=x_2$. Then $p_2(R)\in A_2\cup  B_4\cup B_5 $, otherwise, the robber $R$ will be captured in the round $2$ of the game. Let $p_2(R)=v$ for some $v\in A_2\cup  B_4\cup B_5 $. Next $p_3(D_1)=v_2$ and $p_3(D_2)=u$. Then since $X_1\cup X_3=\emptyset$, by Lemma~\ref{prop}(i) and (ii), $p_3(R)\in \{v_1,v_3\}\subseteq  N(v_2)$ and thus the robber $R$ will be captured in the round $3$ of the game. So $c(G)\leq 2$. So we assume that $p_1(R)=v_3$. Then we choose $p_2(D_1)=v_4$ and $p_2(D_2)=u$. Then $p_2(R)=v_2$ and the rest of the proof is similar to the earlier one. This proves Claim~\ref{unbXemp}. 
	\end{proof}
	
		So by Claim~\ref{unbXnemp} and Claim~\ref{unbXemp}, $B\neq \emptyset$ and without loss of generality, we may assume that $B_1\neq \emptyset$. Let $b\in B_1$. Then we start with $p_1(D_1)=v_1$ and $p_1(D_2)=u$. Then we have $p_1(R)\in \{v_3,v_4\}$. Suppose $p_1(R)=v_3$. Then $p_2(D_1)=v_2$ and $p_2(D_2)=b$. Then by Lemma~\ref{prop}(ii), $p_2(R)\in A_3\cup B_5\cup \{v_2,v_3,v_4\}$. Since $A_3\cup B_5\cup \{v_2,v_3,v_4\}\subseteq N[v_2]\cup N[b]$, the robber $R$ will be captured in the round $2$ of the game. Hence $c(G)\leq 2$.  
		\end{proof}
\begin{proof}[\textbf{Proof of Theorem~\ref{upthcop}}]
  Let $\overline{H}$ be the complement of the Shrikhande graph (see Section~\ref{Shriprop}). Recall that $\overline{H}$ has $16$ vertices and every vertex of $\overline{H}$ has degree  $16-6-1=9$. Also by Theorem~\ref{4k13copeg}, $c(\overline{H})=3$. So we have graph with degree $n-k-1$, where $n=16$ and $k=6$. By Theorem~\ref{cn2char} and Lemma~\ref{upthcomp}, we have $ut($Forb$(4K_1))=6$.
\end{proof}

	\section{The lower threshold degree of Forb$(4K_1)$}\label{LowerThreshold}
	First, we define some terminology used in this section. A maximal connected subgraph of $G$ is called \emph{component} of $G$. A vertex $v$ is called a \emph{cut vertex} of a connected graph $G$ if $G\setminus \{v\}$ is not connected. We now present some lemmas that will be helpful to prove Theorem~\ref{lothcop}.
	

	
	\begin{lemma}\label{4k1mindeg2}
		Let $G$ be a connected $4K_1$-free graph with $\delta(G)=1$. Then $c(G)\leq 2$.
	\end{lemma}
	\begin{proof}
		If $G$ is $P_5$-free, then by Theorem~\ref{p5freecopnum2}, $c(G)\le 2$. So we may assume that $G$ contains a $P_5$ induced by the vertex set $X=\{u_1,u_2,u_3,u_4,u_5\}$ and the edge set $\{u_1u_2,u_2u_3,u_3u_4,u_4u_5\}$. Note that since for any $a\in V(G)\setminus X$, $\{a,u_1,u_3,u_5\}$ does not induce a $4K_1$, every vertex of $V(G)\setminus X$ must be adjacent to at least one vertex from $\{u_1,u_3,u_5\}$. Let $u$ be the vertex with degree $1$. Now we will give a strategy to capture the robber using two cops. Let $D_1$ and $D_2$ be the cops and $R$ be the robber. First suppose $u\in X$. By symmetry, we may assume that $u=u_1$. So $\{u_1\}$ is anticomplete to $V(G)\setminus X$. We start with $p_1(D_1)=u_3$ and $p_1(D_2)=u_5$. Clearly, $p_1(R)\notin X\setminus \{u_1\}$. Also we have $V(G)\setminus X\subseteq N(u_3)\cup N(u_5)$. This implies that $p_1(R)=\{u_1\}$. Next $p_2(D_1)=u_2$ and $p_2(D_2)=u_5$. Thus the robber $R$ will be captured by the cop $D_1$ in the round $2$ of the game.
		
		Finally, we assume that $u\in V(G)\setminus X$. So $\{u\}$ is anticomplete to $(V(G)\setminus X)\setminus \{u\}$. Also $u$ is adjacent to exactly one vertex from $\{u_1,u_3,u_5\}$. Thus we have $\{u\}$ is anticomplete to $\{u_2,u_4\}$. We start with $p_1(D_1)=u_2$ and $p_1(D_2)=u_4$. So $p_1(R)\notin X$. Since for any $a\in  (V(G)\setminus X)\setminus \{u\}$, $\{a,u,u_2,u_4\}$ does not induce a $4K_1$, we have $(V(G)\setminus X)\setminus \{u\}\subseteq N(u_2)\cup N(u_4)$. So $p_1(R)=u$. Now since $N(u)\subseteq N(u_2)\cup N(u_4)$, the robber $R$ will be captured by the cop $D_1$ in the next round of the game.
	\end{proof}	
	
	\begin{lemma}\label{4K1cutvertex}
		Let $G$ be a connected $4K_1$-free graph such that there is a cut vertex in $V(G)$. Then $c(G)\leq 2$.
	\end{lemma}
	\begin{proof}
	  Let $u\in V(G)$ be a cut vertex of $G$. By Lemma~\ref{4k1mindeg2}, we may assume that $\delta(G)\ge 2$. So every component of $G\setminus \{u\}$ contains at least two vertices.  Let $D_1$ and $D_2$ be two cops and $R$ be the robber when the game is played for $G$. First suppose that each component of $G\setminus\{u\}$ is a complete graph. Then we start with $p_1(D_1)=u$. Also for all $j\in \mathbb{N}$, $p_j(D_2)=u$. Let $p_1(R)=v$ and $Q$ be the component of $G\setminus\{u\}$ such that $v\in V(Q)$. Clearly, $uv\notin E(G)$; otherwise the robber $R$ will be captured in the round $1$ of the game. Let $w\in N(u)\cap V(Q)$. Then $p_2(D_1)=w$. Since $u$ is a cut vertex of $G$, $p_2(R)\in V(Q)$. Then since $V(Q)$ is a clique, we have either $p_2(R)=w$ or $wp_2(R)\in E(G)$, and the robber $R$ will be captured in the round $2$ of the game.
	 
	 So we assume that there is a component in $G\setminus\{u\}$ which is not a complete graph, say $Q_1$. Let $V_1=V(Q_1)$ and $V_2=V(G)\setminus (V_1\cup \{u\})$. Since $u$ is a cut vertex of $G$, we have $V_1$ is anticomplete to $V_2$. Let $a,b\in V_1$ such that $ab\notin E(G)$. Since for any two nonadjacent vertices $x,y\in V_2$, $\{a,b,x,y\}$ does not induce a $4K_1$, we have $V_2$ is a clique. Let $a'\in N(u)\cap V_1$ and $b'\in V_2$. By Observation~\ref{obs}, we may assume that $\{a',b'\}$ is not a dominating set. Since $V_2$ is a clique and $V_1$ is anticomplete to $V_2$, there is a vertex in $V_1\setminus N(a')$, say $c$. Now we start with $p_1(D_1)=a'$ and $p_1(D_2)=c$. If $p_1(R)\in V_1$, then since $\{a',b',c,p_1(R)\}$ does not induce a $4K_1$, we have either $a'p_1(R)\in E(G)$ or $cp_1(R)\in E(G)$; thus the robber $R$ will be captured in the round $1$. Since $ua'\in E(G)$, we have $p_1(R)\neq u$. So $p_1(R)\in V_2$. Then $p_2(D_1)=u$ and $p_2(D_2)=c$. So $p_2(R)\in V_2$. Next $p_3(D_2)=c$ and let $p_3(D_1)$ be any vertex in $N(u)\cap V_2$. This implies that the robber $R$ will be captured in the round $3$ of the game. So $c(G)\leq 2$.        
%
	\end{proof}
	\setcounter{claimth}{0}

		\begin{lemma}\label{4K1mindeg3}
			Let $G$ be a connected $4K_1$-free graph such that there is a vertex in $V(G)$ with degree $2$, say $u$. If $c(G\setminus\{u\})\leq 2$, then $c(G)\leq 2$.
		\end{lemma}
		
		\begin{proof}
			By Lemma~\ref{4K1cutvertex}, we may assume that $G$ does not have a cut vertex. Let $G'=G\setminus \{u\}$.	Suppose $c(G')=1$. Then $c(G)\leq c(G')+|\{u\}|=1+1=2$. So we may assume that $c(G')=2$. So the robber can be captured using a set of  two cops in $G'$. Let $D_1$ and $D_2$ be two cops and $R_1$ be the robber when the game is played for $G'$. Let $D_3$ and $D_4$ be two cops and $R_2$ be the robber when the game is played for $G$. Now we will give a strategy to capture the robber using two cops in $G$. We start with $p_1^G(D_3)=p_1^{G'}(D_1)$ and $p_1^G(D_4)=p_1^{G'}(D_2)$. For $i\geq 2$, $p_i^G(D_3)=p_i^{G'}(D_1)$, and $p_i^G(D_4)=p_i^{G'}(D_2)$, if $p_{i-1}^G(R_2)\neq u$.
			
			First suppose that for all $j\in \mathbb{N}$, $p_j^G(R_2)\neq u$. Then by using the strategy to capture $R_1$ by $D_1$ and $D_2$ in $G'$, $R_2$ can be captured by $D_3$ and $D_4$ within a finite round of the game. So we assume that there is a $j\in \mathbb{N}$ such that $p_j^G(R_2)=u$. Let $p_j^G(D_3)=a_1$ and $p_j^G(D_4)=a_2$. Since $d_G(u)=2$, we have $|N_G(u)|=2$. Let $N_G(u)=\{v_1,v_2\}$. Note that $a_1,a_2\notin \{v_1,v_2\}$, otherwise the robber $R_2$ will be captured in the round $j$ of the game.  
			
			Now if every set $U\subseteq V(G)\setminus\{a_1,v_1\}$ such that $U\cup \{a_1,v_1\}$ induces a path and $v_2\in U$, then there does not exist a path between $a_1$ and $v_1$ in $G\setminus\{v_2\}$, then $v_2$ is a cut vertex, a contradiction. Thus there exists a set, say $S_1\subseteq V(G)\setminus\{a_1,v_1,v_2\}$ such that $S_1\cup \{a_1,v_1\}$ induces a path and we choose $|S_1|$ to be the minimal. Similarly we can choose $S_2\subseteq V(G)\setminus\{a_2,v_1,v_2\}$ such that $S_2\cup \{a_2,v_2\}$ induces a path and $|S_2|$ is minimal. Note that if $a_1v_1\in E(G)$, then $S_1=\emptyset$ and if $a_2v_2\in E(G)$, then $S_2=\emptyset$. 
			Next we claim the following:

	\begin{claim}\label{mindegpathprop}
		For $i,k\in \{1,2\}$ and $i\neq k$, $|S_i|\leq 3$ and $S_i\setminus N(a_i)$ is complete to $N(v_i)\setminus  \{u,v_k\}$. 
	\end{claim}
	\begin{proof}[Proof of Claim~\ref{mindegpathprop}]
	We will prove for $i=1$. Assume that $|S_1|\geq 4$. Let $q_1,q_2,\ldots , q_r\in S_1$, for $r\ge 4$ such that $a_1q_1, v_1q_r\in E(G)$. By the choice of $S_1$, we may assume that $a_1q_2,q_2q_r\notin E(G)$. Then $\{a_1,q_2,q_r,u\}$ induces a $4K_1$, a contradiction. So $|S_1|\leq 3$. This proves the first assertion of the claim.
			
	To show the second assertion, suppose to the contrary that there are nonadjacent vertices, say $x\in S_1\setminus N(a_1)$ and $y\in N(v_1)\setminus \{u,v_2\}$.  Since $S_1\cap N(a_1)\neq \emptyset$, we have $|S_1|\geq 2$. If $a_1y\in E(G)$, then since $\{a_1,y,v_1\}$ induces a path, we get a contradiction to the minimality of $|S_1|$. So we have $a_1y\notin E(G)$. Then $\{a_1,x,y,u\}$ induces a $4K_1$, a contradiction. So Claim~\ref{mindegpathprop} holds.
	\end{proof}
	For $i\in \{1,2\}$, we choose the following until for some $k$ satisfying $p_{k}(R_2)\in N[p_{k}(D_3)]\cup N[p_{k}(D_4)]$:
	\begin{enumerate}[label=(\roman*)]
		\item If $|S_i|=3$, then  $p_{j+1}(D_{3+i-1})\in S_i\cap N(a_i)$, $p_{j+2}(D_{3+i-1})\in S_i\setminus(N(a_i)\cup N(v_i))$, $p_{j+3}(D_{3+i-1})\in S_i\cap N(v_i)$, and for $\ell \geq j+4$, $p_{\ell}(D_{3+i-1})=v_i$.
		\item If $|S_i|=2$, then $p_{j+1}(D_{3+i-1})\in S_i\cap N(a_i)$, $p_{j+2}(D_{3+i-1})\in S_i\cap N(v_i)$, and for $\ell \geq j+3$, $p_{\ell}(D_{3+i-1})=v_i$.
		\item If $|S_i|=1$, then $p_{j+1}(D_{3+i-1})\in S_i\cap N(a_i)$, and for $\ell \geq j+2$, $p_{\ell}(D_{3+i-1})=v_i$.
		\item If $|S_i|=0$, then for $\ell \geq j+1$, $p_{\ell}(D_{3+i-1})=v_i$.
	\end{enumerate}
			
			Next we claim the following:
			
	\begin{claim}\label{robcap}
		If for some $k\geq j+1$, the robber $R_2$ is not captured in the round $k$, then $p_{k}(R_2)\in \{u,v_1,v_2\}$.
	\end{claim}
		\begin{proof}[Proof of Claim~\ref{robcap}]
		For the sake of contradiction, assume that $d=p_{k}(R)$ and $d\notin \{u,v_1,v_2\}$. Since $p_j(R_2)=u$ and $N(u)=\{v_1,v_2\}$, the robber $R_2$ is required at least two rounds to reach $d$ from $u$. So we have $k\geq j+2$. Without loss of generality, we may assume that $p_{j+1}(R_2)=v_1$. Since $v_1p_{j+1}(D_3)\notin E(G)$, we have $|S_1|\geq 2$. Since $p_{j+2}(D_3)\in S_1\setminus N(a_1)$, by Claim~\ref{mindegpathprop}, $\{p_{j+2}(D_3)\}$ is complete to $N(v_1)\setminus \{u,v_2\}$, we have $p_{j+2}(R_2)\in \{u,v_1,v_2\}$. First suppose, $p_{j+2}(R_2)=v_1$. Then $|S_1|=3$. Thus by Claim~\ref{mindegpathprop}, $p_{j+3}(R_2)\in \{u,v_1,v_2\}$, and since by our choice, either $v_1p_{j+3}(D_3),v_2p_{j+3}(D_4)\in E(G)$ or $p_{j+4}(D_3)=v_1$ and $p_{j+4}(D_4)=v_2$, we have the robber $R_2$ will be captured within the round $j+4$ of the game, a contradiction. So we assume that $p_{j+2}(R_2)\in \{u,v_2\}$. Then again similar to the previous argument, we have the robber $R_2$ will be captured before reaching the vertex $d$, a contradiction. This proves Claim~\ref{robcap}.
	\end{proof}		
			
			Now by our choice and by Claim~\ref{robcap}, the robber $R_2$ will be captured within the round $j+4$ of the game.
		\end{proof}

	\begin{proof}[\textbf{Proof of Theorem~\ref{lothcop}}] 
		By Theorem~\ref{4k13copeg}, we have $c($Forb$(4K_1))=3$. So by Lemma~\ref{4k1mindeg2} and Lemma~\ref{4K1mindeg3}, we have every connected $4K_1$-free graph $G$ satisfying $c(G)=3$ has minimum degree at least $3$. So $lt(G)\geq 3$.
	\end{proof}
	
	\section{Computational results}\label{compute}
	By Theorem~\ref{4k13copeg}, there is a connected $4K_1$-free graph with the cop number $3$ that has $16$ vertices. This naturally leads to the following question.
		\begin{quest}\label{quest}
		Can we have a $4K_1$-free graph with the number of vertices less than $16$ and the cop number $3?$
	\end{quest} 
	Now we give partial answer to Question~\ref{quest} computationally. In this regards, we also discuss  advantages of computationally finding $c({\cal G})$ and  graphs in ${\cal G}$ which satisfy $c({\cal G})$ using the upper threshold degree and the lower threshold degree, for a given class of graphs ${\cal G}$.  From a result of Turcotte \cite{Turcotte21}, every graph with the cop number $3$ and vertex set of size $11$ contains a $4K_1$. Thus we have that if a $4K_1$-free graph $G$ has the cop number $3$, then $G$ has at least 12 vertices. 
	
Next we will show computationally that a $4K_1$-free graph with $12$ vertices has the cop number $2$. For that we use nauty \cite{Mckay13} in ``SageMath" to generate  graphs with $12$ vertices. Let $\mathcal{F}:=\{G\in \mathcal{F} | |V(G)|=12, \Delta(G)\le 5, \text{and~} G \text{~is connected}\} $

	\begin{table}[h]
		\centering
		\begin{tabular}{|c|c|}
			\hline
			$\delta(G)$~$\forall~G\in \mathcal{F}$ & Number of graphs \\ \hline
			$\ge1 $           & 471142472    \\ \hline
			$\ge2 $           & 333204830    \\ \hline
			$\ge3 $           & 96689615     \\ \hline
		\end{tabular}
		\caption{Number of connected graphs having $12$ vertices and maximum degree at most $5$.}
		\label{Totgraph}
	\end{table}
	
	By Theorem~\ref{upthcop} and Theorem~\ref{lothcop}, $ut(\text{Forb}(4K_1))=6$ and $lt(\text{Forb}(4K_1))\ge 3$. Therefore if there exists a $4K_1$-free graph having $12$ vertices with the cop number $3$, then the maximum degree of that graph is at most $5$ and the minimum degree of that graph is at least $3$. From Table~\ref{Totgraph}, it is evident that if we use Theorem~\ref{lothcop}, then we need to search very less number of graphs to  find a $4K_1$-free graph with the cop number $3$. Thus we only look for $4K_1$-free graphs having $12$ vertices with maximum degree at most $5$ and minimum  degree at least $3$. There are   such $96689615$ graphs. Out of which $31155$ graphs are $4K_1$-free. 
	
	By Observation~\ref{obs}, it is enough to check the graph with the domination number $3$. Note that there are $24455$ number of $4K_1$-free connected graphs with $12$ vertices and the domination number  $3$. Now by Theorem~\ref{p5freecopnum2}, it is enough to check the graph having an induced $P_5$. Since there are only one $(P_5,4K_1)$-free connected graph which has $12$ vertices and the domination number $3$, we only need to check $24454$ number of graphs. Next we have split this case as the graph having an induced $C_7$, $C_7$-free having an induced $C_6$, and $(C_7,C_6)$-free. We have used the algorithm mentioned in \cite{Turcotte21} to obtain the cop number of a given graph. We have implemented all the codes in SageMath\footnote{For detail codes and results see \url{https://sites.google.com/site/dpradhan1/cops-and-robber}}. Now we have the following results for graphs having $12$ vertices with $\delta(G)\ge 3$ and $\Delta(G)\le 5$:
	\begin{enumerate}[label=(\roman*)]
		\item Number of $4K_1$-free graphs with induced $C_7$ and the domination number $3$ is $10067$.
		\item Number of $(4K_1,C_7)$-free graphs with induced $C_6$ and the domination number $3$ is $12021$.
		\item Number of $(4K_1,C_7,C_6)$-free graphs with induced $P_5$ and the domination number $3$ is $2366$.
	\end{enumerate}
	Finally, we have obtained that there does not exist any $4K_1$ graph having $12$ vertices with the cop number $3$. 
	 
	\section{The cop number of $(pK_1+qK_2)$-free graphs}\label{pk1qk2res}
	
	In this section, we have studied $c($Forb$(pK_1+K_2))$ and extended this result for the class of $(pK_1+qK_2)$-free graphs. To prove the above result, we utilize Turcotte's result~\cite{Turcotte22}, which is stated below.
	\begin{theorem}[\cite{Turcotte22}] \label{2k1k2-free}
		If $G\in $ \textup{Forb}$(2K_1+ K_2)$, then $c(G)\le 3$.
	\end{theorem}
	\begin{theorem}[\cite{Turcotte22}] \label{lk2-free}
		If $G\in $ \textup{Forb}$(q K_2)$, then $c(G)\le 2q-2$.
	\end{theorem}
	\begin{theorem}\label{mk1k2}
		If $G\in$ \textup{Forb}$(pK_1+ K_2)$ for $p\ge 2$, then $c(G)\le p+1$. 
	\end{theorem}
	\begin{proof}
		To prove the theorem, we use the induction on $p$. The base case of the induction hypothesis that is $p=2$, follows from Theorem~\ref{2k1k2-free}. Now we assume that $p> 2$ and the result is true for $p-1$.
		
		Let $G$ be a $pK_1+ K_2$-free graph. If $G$ is a $(p-1)K_1$-free graph, then $G$ is a $(p-1)K_1+ K_2$-free graph; thus by the induction hypothesis, $c(G)\le p< p+1$. So we assume that $G$ contains an induced $(p-1)K_1$. Let $A\subseteq V(G)$ be an independent set such that $|A|=p-1$ and such a set exists since $G$ contains an induced $(p-1)K_1$. Also, let $A':= \{v\in V(G)\setminus A| N(v)\cap A=\emptyset\}$. Since $G$ is a connected $pK_1+ K_2$-free graph, $G[A']$ is a $K_1+K_2$-free graph. Note that every $K_1+K_2$-free graphs are complete multipartite graphs. Let $u,v$ be two vertices from two distinct partite sets of $G[A']$. Then $A\cup \{u,v\}$ is a dominating set of $G$. Hence $\gamma(G)\le |A|+2 \le p+1$. Now by using Observation~\ref{obs}, we have  $c(G)\le \gamma(G)\le p+1$. This completes the proof of Theorem~\ref{mk1k2}.	
	\end{proof}
	\begin{theorem}\label{mk1lk2}
		If $G\in$ \textup{Forb}$(pK_1+ qK_2)$ for $p\ge 2$ and $q\ge 2$, then $c(G)\le p+2q-2$. 
	\end{theorem}
	\begin{proof}
		Suppose that $G$ contains a $pK_1$ induced by the vertex set $L=\{v_1,v_2,\ldots,v_p\}$. Also, let $L':= \{v\in V(G)\setminus L| N(v)\cap L=\emptyset\}$. Then $G[L']$ is $qK_2$-free. Throughout the game we will put $p$ cops, say $C_1,C_2,\ldots,C_p\in L$, such that $v_j$ is the position of cop $C_j$, for all $j\in \{1,2,\ldots,p\}$. So the robber will be move in $G[L']$. By Theorem~\ref{lk2-free}, it is enough to require $2q-2$ cops in $G[L']$ to capture the robber in a finite round of the game. Thus $c(G)\leq p+2q-2$. So we assume that $G$ is $pK_1$-free. Then $\gamma(G)\leq p-1$. Hence by Observation~\ref{obs}, $c(G)\leq p-1$. Since $q\geq 1$, we have $c(G)\leq p+2q-2$. 
		This completes the proof of Theorem~\ref{mk1lk2}.	
	\end{proof}
	
	\section{Future direction}\label{future}
	We have proved that for Forb$(4K_1)$, $c($Forb$(4K_1))=3$. So Problem~\ref{pk1prob} is still open for $p\ge 5$. We also showed a $4K_1$-free graph having the cop number $3$ has at least $13$ vertices and proved the existence of such graph with $16$ vertices.  This leads us to pose the following question.
	\begin{quest}
		Does there exist a $4K_1$-free graph with $n$-vertices having the cop number $3$ where $n\in\{13,14,15\}?$
	\end{quest}
	We also use the upper threshold degree and the lower threshold degree to prove our results. We believe that the upper threshold degree and the lower threshold degree will help us to find the cop number for various other class of graphs. Indeed, we propose the following question.
	\begin{quest}
		For a given class of graphs ${\cal G}$,  find the $ut({\cal G})$ and $lt({\cal G})$.
	\end{quest}

	\section*{Declarations}
	
	\noindent{\bf Conflict of interest} The authors do not have any financial or non financial interests that are directly or indirectly related to the work submitted for publication.
	
	
	\noindent{\bf Acknowledgments} The first author acknowledges the Indian Institute of Technology (ISM), Dhanbad, India, for hosting a portion of this research. Also, the research of the first author is supported in part Japan Science and Technology Agency (JST) as part of Adopting Sustainable Partnerships for Innovative Research Ecosystem (ASPIRE), Grant Number JPMJAP2302. The research of the third author is suppported by  Anusandhan National Research Foundation (ANRF) under Core Research Grant Number CRG/2023/003749.

\end{document}